\documentclass[11 pt]{article}
\usepackage{amsfonts,amsmath,color,amsthm,amssymb, url, enumerate, bbm}
\usepackage[margin=1in]{geometry}

\usepackage{tikz, etoolbox}
\usetikzlibrary{shapes}
\usetikzlibrary{arrows}

\newcommand{\f}{\frac}

\newcommand{\Z}{\mathbb{Z}}

\newcommand{\R}{\mathbb{R}}

\definecolor{purple}{RGB}{128,0,128}

\def\qed{\vbox{\hrule\hbox{\vrule\kern3pt\vbox{\kern6pt}\kern3pt\vrule}\hrule}}

\newtheorem{thm}{Theorem}[section]
\newtheorem{lm}[thm]{Lemma}

\newtheorem{rem}[thm]{Remark}
\newtheorem{prop}[thm]{Proposition}

\newtheorem*{thm*}{Theorem}
\newtheorem*{cor*}{Corollary}
\newtheorem*{lm*}{Lemma}
\newtheorem*{prop*}{Proposition}

\theoremstyle{definition}

\newtheoremstyle%
 {Aside}%
 {}{}%
 {\color{purple}\itshape}
 {}%
 {\color{purple}\bfseries}%
 {\color{purple}.}%
 { }{}
\theoremstyle{Aside}

\title{Subtour Elimination Constraints Imply a Matrix-Tree Theorem SDP Constraint for the TSP}
\author{Samuel C. Gutekunst and David P. Williamson}
\date{}

\begin{document}
\maketitle

\begin{abstract}
De Klerk, Pasechnik, and Sotirov \cite{Klerk08} give a semidefinite programming constraint for the Traveling Salesman Problem (TSP) based on the matrix-tree Theorem.  This constraint says that the aggregate weight of all spanning trees in a solution to a TSP relaxation is at least that of a cycle graph.  In this note, we show that the semidefinite constraint holds for any weighted 2-edge-connected graph and, in particular, is implied by the subtour elimination constraints of the subtour elimination linear program. Hence, this semidefinite constraint is implied by a finite set of linear inequality constraints.
\end{abstract}

\section{Introduction and The Matrix-Tree Theorem} 

The Traveling Salesman Problem (TSP) is a fundamental problem in combinatorial optimization and a canonical NP-hard problem.  Efficiently computable relaxations of the TSP are used to find optimal and near-optimal TSP solutions, and recently, several relaxations based on semidefinite programs (SDPs) have been proposed (see, e.g., Cvetkovi{\'c}, {\v{C}}angalovi{\'c}, and Kova{\v{c}}evi{\'c}-Vuj{\v{c}}i{\'c} \cite{Cvet99}, de Klerk, Pasechnik, and Sotirov \cite{Klerk08}, and de Klerk and Sotirov \cite{Klerk12b}).  

A common source of SDP constraints for the TSP is spectral graph theory: the SDP of Cvetkovi{\'c}  et al.\  \cite{Cvet99} is based on algebraic connectivity, and de Klerk et al.\ \cite{Klerk08} give a constraint based on  Kirchoff's matrix-tree theorem.  Goemans and Rendl \cite{Goe00}  show that the constraints used in the SDP relaxation of Cvetkovi{\'c} et al.\ \cite{Cvet99} are implied by the canonical TSP relaxation, the subtour elimination linear program (see Equation (\ref{eq:SLP}) below for the precise definition of this linear program). In this note we show that the matrix-tree theorem constraint of de Klerk et al.\ \cite{Klerk08} is also implied by the subtour elimination linear program constraints.

The matrix-tree theorem dates back to the mid-19th century (Kirchoff \cite{Kir47})  and connects the number of spanning trees of a graph to the Laplacian matrix of that graph. Let $G=(V, E)$ be a simple, undirected graph, and suppose each edge $e$ has weight $x_e\geq 0.$  Let $X$ be the corresponding weighted adjacency matrix, so that $X$ has zero diagonal and $X_{ij}=X_{ji}=x_{\{i, j\}}.$  The {\bf Laplacian} of $X$ is the $|V|\times |V|$ matrix $L(X)$ defined entrywise as $$L(X)_{i, j} = \begin{cases} -x_e,& \{i, j\}\in E \\ \sum_{e: e\cap i \neq \emptyset} x_e ,& i=j \\ 0,& \text{ else}. \end{cases}$$ Suppose that $\mathcal{T}_G$ is the set of spanning trees of $G$.  The matrix-tree theorem is the remarkable result that  any \emph{principal minor} of $L(X)$ (i.e., the determinant of the matrix obtained by removing the $i$th row and column of $L(X)$ for any $1\leq i\leq |V|$) equals $\sum_{T\in \mathcal{T}_G} \prod_{e\in T} x_e.$  In the case that $x_e=1$ for every edge in $G$, the term $\sum_{T\in \mathcal{T}_G} \prod_{e\in T} x_e$ counts the number of spanning trees of $G$.  See Theorem VI.29 in \cite{Tut01}, e.g., for a proof of this general version of the matrix-tree theorem.

De Klerk  et al.\ \cite{Klerk08} notice that any Hamiltonian cycle on $n$ vertices has $n$ spanning trees (delete any individual edge).  They use the matrix-tree theorem to derive a constraint for SDP relaxations of the TSP saying that ``the aggregate weight of spanning trees is at least $n$.''  We show that this constraint is implied by constraints in the subtour elimination linear program:
\begin{thm}\label{thm:Main1}
Let $x\in \R^E$ be a feasible solution to the subtour LP (\ref{eq:SLP}) and let $G$ be the complete graph.  Let $X$ be the symmetric matrix where $X_{ij}=X_{ji}=x_{\{i, j\}}$ and $X_{ii}=0$ for all $i$.  Then $X$ satisfies the matrix-tree theorem constraint:
$$\sum_{T\in \mathcal{T}_{G}} \prod_{e\in T} x_e \geq n.$$
\end{thm}
Our results show that the matrix-tree theorem constraint (requiring linear matrix inequalities) of  de Klerk et al.\ \cite{Klerk08}  is  weaker than the subtour LP (using just linear inequalities). They can also be stated more generally: any graph $G$ that is 2-edge-connected in a weighted sense (i.e. $\sum_{e:|e\cap S|=1}x_e \geq 2$ for every $\emptyset \subsetneq S \subsetneq V$) satisfies $\sum_{T\in \mathcal{T}_G} \prod_{e\in T} x_e \geq n.$  Our result follows from a theorem of  Ok and Thomassen \cite{Ok17} that lower-bounds the number of spanning trees in an unweighted, loopless, undirected multigraph.  In Section \ref{sec:bg}, we provide background on the TSP and relaxations.  In Section \ref{sec:proof} we then state the theorem from Ok and Thomassen \cite{Ok17} and use it to deduce Theorem \ref{thm:Main1}.

\section{Preliminaries}\label{sec:bg}

The   Traveling Salesman Problem can be stated as follows.  Let $G=K_n$ be the complete graph on $V=[n]:=\{1, 2, ..., n\}.$  For each $e=\{i, j\}\in G$, associate an edge cost $c_e$  (interpreted  as the cost of traveling from vertex $i$ to vertex $j$ or vice versa).  The TSP is to find a minimum-cost tour on $G$ visiting every vertex exactly once, i.e., to find a minimum-cost Hamiltonian cycle on $K_n$ with respect to the edge costs.

The prototypical TSP relaxation is the the subtour elimination linear program (also referred to as the Dantzig-Fulkerson-Johnson relaxation \cite{Dan54} and the Held-Karp bound \cite{Held70}, and which we will refer to as the {\bf subtour LP}).   For $S\subset V$, denote the set of edges with exactly one endpoint in $S$ by $\delta(S):=\{e=\{i, j\}: |\{i, j\}\cap S|=1\}$ and let $\delta(v):=\delta(\{v\}).$  For $F\subset E,$ let $x(F)$ denote the sum of $x$ over those edges in $F$: $x(F)=\sum_{e\in F} x_e.$  The subtour LP is:
\begin{equation}\label{eq:SLP}\begin{array}{l l l}
\min & \sum_{e\in E} c_e x_e & \\
\text{subject to} & x(\delta(v)) = 2, & v=1, \ldots, n \\
& x(\delta(S)) \geq 2, & S\subset V: S\neq \emptyset, S\neq V \\
&0\leq x_e \leq 1, & e\in E. \end{array}
\end{equation}
The subtour LP is a {\bf relaxation} of the TSP because 1) every Hamiltonian cycle has a corresponding feasible solution to the subtour LP, and 2) the value of the subtour LP for such a feasible solution equals the cost of the corresponding Hamiltonian cycle.   

Significant recent research has gone into developing relaxations instead based on semidefinite programs (SDPs).  See, e.g., Cvetkovi{\'c}, {\v{C}}angalovi{\'c}, and Kova{\v{c}}evi{\'c}-Vuj{\v{c}}i{\'c} \cite{Cvet99} (who both introduce an SDP relaxation based on algebraic connectivity), de Klerk, Pasechnik, and Sotirov \cite{Klerk08} (who introduce an SDP relaxation based on the theory of association schemes and give the matrix-tree theorem-based SDP constraint), and  de Klerk and Sotirov \cite{Klerk12b} (who use symmetry reduction to strengthen the SDP of  de Klerk et al.\ \cite{Klerk08}).  Various results have characterized the performance of these SDPs (Goemans and Rendl \cite{Goe00}, Gutekunst and Williamson \cite{Gut17}, and Gutekunst and Williamson \cite{Gut19}).

The TSP SDP relaxations generally have some symmetric matrix variable $X$ that can be interpreted as a weighted adjacency matrix. There are typically constraints enforcing that $X_{ii}=0,$ and since $X$ is symmetric, $X_{ij}=X_{ji}$ can be thought of as the weight $x_e$ on edge $e=\{i, j\}$.   In a feasible solution to the SDP relaxation taking on integral values, constraints force $X$ to be the weighted adjacency matrix of a Hamiltonian cycle.  There are generally constraints that directly enforce that $X$ is 2-regular in a weighted sense: every row of $X$ sums to 2, in analogy to the subtour LP constraints that $x(\delta(v))=2$.  If $G$ is the graph with edge weights $x_e,$ the constraints imply that the corresponding Laplacian matrix to $X$ is $L(X)=2I-X$.  Throughout we treat $X$ as the weighted adjacency matrix of a complete graph $G=K_n$ where edges can have weight zero.\footnote{Any spanning tree $T$ containing an edge of weight zero has $ \prod_{e\in T} x_e =0$ and doesn't contribute to $\sum_{T\in \mathcal{T}_{G}} \prod_{e\in T} x_e \geq n$.  We can let $G=K_n$ without loss of generality, as any other graph can be extended to the complete graph by placing a weight of zero on all missing edges; the weighted adjacency matrix, Laplacian, and aggregate spanning tree weight $ \sum_{T\in \mathcal{T}_{G}} \prod_{e\in T} x_e \geq n$ will not change.}

  Let $A_{-i}$ denote the matrix obtained by deleting the $i$th row and column of $A$. Note that $2I-X$ is positive semidefinite, so $(2I-X)_{-i}$ is positive semidefinite for all $i$.   De Klerk et al.\ \cite{Klerk08}'s observation that a Hamiltonian cycle has $n$ spanning trees, together with the aforementioned matrix-tree theorem, allows them to introduce the SDP constraint  \begin{equation} \label{eq:MTT} \text{det}\left(\left(2I-X\right)_{-1}\right) \geq n.\end{equation}  Since the set $\{Z\succeq 0: \text{det}(Z) \geq n\}$ can be expressed as a linear matrix inequality (see, e.g., Section 3.2 of Nemirovskii \cite{Nem05}),  the constraint in  Equation (\ref{eq:MTT}) can be written as a linear matrix inequality for use in TSP SDP relaxations.  De Klerk  et al.\ \cite{Klerk08} note  that this constraint strengthens a semidefinite programming relaxation of the TSP from Cvetkovi{\'c}, {\v{C}}angalovi{\'c}, and Kova{\v{c}}evi{\'c}-Vuj{\v{c}}i{\'c} \cite{Cvet99}.
 We refer to Equation (\ref{eq:MTT}) as the ``matrix-tree theorem constraint.''

\section{The Matrix-Tree Theorem Constraint}\label{sec:proof}

To prove our main result, we will use the following result from Ok and Thomassen \cite{Ok17} which relates edge-connectivity to spanning trees. An unweighted, undirected, loopless multigraph $G=(V, E)$ is {\bf $k$-edge-connected} if $G$ is still connected after the removal of any $k-1$ edges.

\begin{thm}[Theorem 1 in Ok and Thomassen \cite{Ok17}]\label{lem:OK}
Let $G$ be an weighted, loopless, undirected multigraph that is $k$-edge-connected.  Then $G$ has at least $n\left(\f{k}{2}\right)^{n-1}$ spanning trees.
\end{thm} 
We first use it to prove the following:
\begin{prop}\label{thm:MT2}
Let $G=(V, E)$ be a weighted simple graph with rational edge weights given by $x\in\R^{E}.$   If $x$ is an extreme point of the subtour LP (\ref{eq:SLP}), then $$\sum_{T\in\mathcal{T_G}} \prod_{e\in T} x_e \geq n.$$  
\end{prop}
\noindent Theorem \ref{thm:Main1} will then follow as an immediate consequence.  

To prove Proposition \ref{thm:MT2}, we start with a symmetric, simple weighted graph $G=(V, E)$ with edge weights given by $x\in\R^E.$  Because $x$ is rational, we will be able to scale $x$ so that $Rx\in \Z^E$.  Then we let $G'=(V, E')$ be an undirected, loopless, unweighted multigraph with $Rx_e$ copies of edge $e$.  Moreover, if $x(\delta(S))\geq 2$ then $Rx(\delta(S))\geq 2R$ so that $G'$ will be $2R$-edge-connected.  We can then appeal to Theorem \ref{lem:OK}, find a large number of spanning trees, and find corresponding spanning trees in $G$.  

We first verify that the aggregate weight of spanning trees in $G'$ (as an unweighted multigraph with $Rx_e$ copies of edge $e$) matches the aggregate weight of spanning trees in $G$ (as a weighted simple graph where edge $e$ has weight $Rx_e$).  To do so, we  apply the following lemma iteratively.

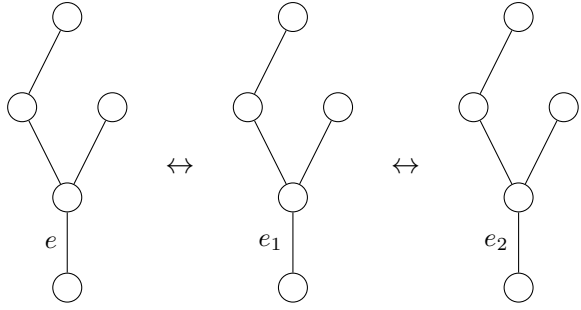
\begin{figure}[t]
\centering

\begin{tikzpicture}[scale=0.6]
\tikzset{vertex/.style = {shape=circle,draw,minimum size=1em}}
\tikzset{edge/.style = {->,> = latex'}}
\tikzstyle{decision} = [diamond, draw, text badly centered, inner sep=3pt]
\tikzstyle{sq} = [regular polygon,regular polygon sides=4, draw, text badly centered, inner sep=3pt]
% vertices
\node[vertex] (a) at  (-9, 0) {};
\node[vertex] (b) at  (-7, 0) {};
\node[vertex] (c) at  (-8, 2) {};
\node[vertex] (d) at  (-8, -2) {};
\node[vertex] (e) at  (-8, -4) {};

\node[vertex] (a1) at  (-4, 0) {};
\node[vertex] (b1) at  (-2, 0) {};
\node[vertex] (c1) at  (-3, 2) {};
\node[vertex] (d1) at  (-3, -2) {};
\node[vertex] (e1) at  (-3, -4) {};

\node[vertex] (a2) at  (1, 0) {};
\node[vertex] (b2) at  (3, 0) {};
\node[vertex] (c2) at  (2, 2) {};
\node[vertex] (d2) at  (2, -2) {};
\node[vertex] (e2) at  (2, -4) {};

\draw  (-5.5, -1) node[below] {$\leftrightarrow$};
\draw  (-.5, -1) node[below] {$\leftrightarrow$};
%edges
 \path[every node/.style={font=\sffamily\small}]
        (d) edge node [left] {$e$} (e);
        
         \path[every node/.style={font=\sffamily\small}]
        (d1) edge node [left] {$e_1$} (e1);
        
         \path[every node/.style={font=\sffamily\small}]
        (d2) edge node [left] {$e_2$} (e2);

\draw (c) -- (a);
\draw (d) -- (a);
\draw (d) -- (b);

\draw (c1) -- (a1);
\draw (d1) -- (a1);
\draw (d1) -- (b1);

\draw (c2) -- (a2);
\draw (d2) -- (a2);
\draw (d2) -- (b2);

\end{tikzpicture}\caption{A sample tree instantiation in $\mathcal{T}_G^e, \mathcal{T}_{G'}^1,$ and $\mathcal{T}_{G'}^2$}\label{fig:bij}
\end{figure}

\begin{lm}\label{lm:rat}
Let $G$ be a weighted loopless multigraph.  Let $e=\{u, v\} \in G$ and let $G'$ be obtained from $G$ by splitting $e$ into two copies $e_1=e_2=\{u, v\}$ and assigning nonnegative weights $x'$ to the edges in $G'$ so that $x_e=x'_{e_1}+x'_{e_2}$ (and $x_f=x'_f$ for all other edges $f$ in $G$). Then 
$$\sum_{T\in\mathcal{T}_G} \prod_{f\in T} x_f=\sum_{T\in\mathcal{T}_{G'}} \prod_{f\in T} x'_f.$$
\end{lm}

In the proof, we use $\sqcup$ to denote a partition: $S=A\sqcup B$ means $S=A\cup B$ and $A\cap B=\emptyset.$  We also use $\backslash$ for set-minus, so that $S\backslash A=\{x\in S: x\notin A\}.$

\begin{proof}
This result follows by partitioning $\mathcal{T}_{G'}.$  No $T\in \mathcal{T}_{G'}$ can contain both $e_1$ and $e_2$ so we write $$\mathcal{T}_{G'}=\mathcal{T}^0_{G'} \sqcup \mathcal{T}^1_{G'}\sqcup\mathcal{T}^2_{G'}$$ where $\mathcal{T}^i_{G'}$ consists of those spanning trees including edge $i$ for $i=1, 2$ and  $\mathcal{T}^0_{G'}$ consists of those trees  including neither $e_1$ nor $e_2.$
We analogously partition $$\mathcal{T}_G=\mathcal{T}^0_{G} \sqcup \mathcal{T}^e_{G},$$ where $\mathcal{T}^0_{G}$ consists of spanning trees not using $e$ and $\mathcal{T}^e_{G}$ consists of spanning trees using $e.$  

The trees in  $\mathcal{T}^1_{G'},$ $\mathcal{T}^2_{G'},$ and $\mathcal{T}^e_{G}$ all use exactly one edge $\{u, v\}$ (and other than using exactly one of $e_1, e_2,$ or $e$ as $\{u, v\}$, use the exact same other edges).  Hence if $T\in \mathcal{T}^1_{G'}$, then  $\left(T\backslash e_1\right)\cup e_2\in \mathcal{T}^2_{G'}$ and $\left(T\backslash e_1\right)\cup e \in \mathcal{T}^e_{G}.$  This process gives a one-to-one correspondence between trees in $\mathcal{T}^1_{G'},$ $\mathcal{T}^2_{G'},$ and $\mathcal{T}^e_{G};$ see Figure \ref{fig:bij}.  Analogously, $\mathcal{T}^0_{G'}=\mathcal{T}^0_G.$  Hence:
\begin{align*}
\sum_{T\in\mathcal{T}_{G'}} \prod_{f\in T} x'_f 
&=\sum_{T\in\mathcal{T}^0_{G'}} \prod_{f\in T} x'_f+\sum_{T\in\mathcal{T}^1_{G'}} \prod_{f\in T} x'_f+\sum_{T\in\mathcal{T}^2_{G'}} \prod_{f\in T} x'_f \\
&=\sum_{T\in\mathcal{T}^0_{G}} \prod_{f\in T} x_f+(x'_{e_1}+x'_{e_2})\sum_{T\in\mathcal{T}^1_{G'}} \prod_{f\in T, f\neq e_1} x'_f \\
&=\sum_{T\in\mathcal{T}^0_{G}} \prod_{f\in T} x_f+x_e\sum_{T\in\mathcal{T}^1_{G}} \prod_{f\in T, f\neq e} x_f \\
&=\sum_{T\in\mathcal{T}^0_{G}} \prod_{f\in T} x_f+\sum_{T\in\mathcal{T}^1_{G}} \prod_{f\in T} x_f \\
&=\sum_{T\in\mathcal{T}_G} \prod_{f\in T} x_f.
\end{align*}
\hfill
\end{proof}

We now prove our main theorem in the special case of subtour LP extreme points.

\begin{figure}[t]
\centering

\begin{tikzpicture}[scale=0.6]
\tikzset{vertex/.style = {shape=circle,draw,minimum size=1em}}
\tikzset{edge/.style = {->,> = latex'}}
\tikzstyle{decision} = [diamond, draw, text badly centered, inner sep=3pt]
\tikzstyle{sq} = [regular polygon,regular polygon sides=4, draw, text badly centered, inner sep=3pt]
% vertices
\node[vertex] (a) at  (1, 0) {};
\node[vertex] (b) at  (0, 1) {};
\node[vertex] (c) at  (0, -1) {};
\node[vertex] (d) at  (3, 0) {};
\node[vertex] (e) at  (4, 1) {};
\node[vertex] (f) at  (4, -1) {};
\draw  (2, -2) node[below] {$G$};

\node[vertex] (a1) at  (8, 0) {};
\node[vertex] (b1) at  (7, 1) {};
\node[vertex] (c1) at  (7, -1) {};
\node[vertex] (d1) at  (10, 0) {};
\node[vertex] (e1) at  (11, 1) {};
\node[vertex] (f1) at  (11, -1) {};
\draw  (9, -2) node[below] {$G'$};

%edges   
\draw[dotted,line width=1.2pt] (a) -- (b);
\draw[dotted,line width=1.2pt] (b) -- (c);
\draw[dotted,line width=1.2pt] (c) -- (a);

\draw[dotted,line width=1.2pt] (d) -- (e);
\draw[dotted,line width=1.2pt] (e) -- (f);
\draw[dotted,line width=1.2pt] (f) -- (d);

\draw (a) -- (d);
\draw (b) -- (e);
\draw (c) -- (f);

\draw (a1) -- (b1);
\draw (b1) -- (c1);
\draw (c1) -- (a1);
\draw(d1) -- (e1);
\draw (e1) -- (f1);
\draw (f1) -- (d1);

\path (a1) edge [bend left] (d1);
\path (d1) edge [bend left] (a1);

\path (b1) edge [bend left=15] (e1);
\path (e1) edge [bend left=15] (b1);

\path (c1) edge [bend left=15] (f1);
\path (f1) edge [bend left=15] (c1);

\end{tikzpicture}\caption{The left shows a simple, weighted graph $G$ where dashed edges have weight $1/2$ and full edges have weigh $1.$  In this case $R=2$ and the right shows the corresponding unweighted multigraph $G'.$}\label{fig:R}
\end{figure}
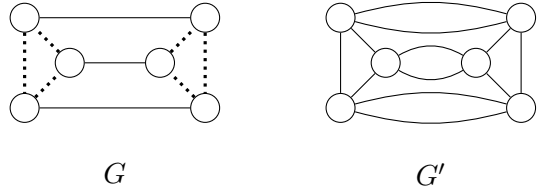

\begin{proof}[Proof (of Proposition \ref{thm:MT2})]
Let $x$ be a feasible extreme point of the subtour LP.  Then  $x(\delta(S))\geq 2$ for all $S$ with $1\leq |S| \leq |V|-1$ and, moreover, it is well-known that $x$ is rational.\footnote{Extreme points occur where a certain number of constraints hold with equality.  Cramer's rule, e.g., shows that if the constraints of a linear program have rational coefficients, then every extreme point is rational.  This is the case for the subtour LP.}

We first convert $G$ into a loopless, unweighted multigraph.  To do so, suppose that $x_{e_i}=\f{s_i}{r_i}$ in lowest terms.  Let $$R=\mathcal{LCM}(r_1, ..., r_m).$$  Let $G'$ denote the graph $G$ with weights $x'_e=Rx_e$ for all $e\in G$ and let $X'=RX$.  Then we make two observations: $x'_e\in\Z$ for all $e$, and properties of deteriminants imply \begin{equation}\label{eq:det1} \text{det}\left(\left(L(X')_{-1}\right)\right) = R^{n-1}\text{det}\left(\left(L(X)_{-1}\right)\right).\end{equation}  To compute $\text{det}\left(\left(L(X')_{-1}\right)\right)$ we appeal to the matrix-tree theorem; by Lemma \ref{lm:rat} it is equivalent (in terms of the aggregate weight of spanning trees) to view $G'$ as a loopless unweighted multigraph where there are $x'_e$ copies of edge $e$ (so that there are $s_i \f{R}{r_i}\in \Z$ copies of edge $e_i$, each of weight 1).   See Figure \ref{fig:R}.

Note that $x(\delta(S))\geq 2$ implies that $x'(\delta(S))\geq 2R.$  Thus $G'$ is $2R$-edge-connected and by Theorem \ref{lem:OK}, $G'$ has at least $nR^{n-1}$ spanning trees; since every edge of $G'$ has weight 1, the matrix-tree theorem states
$$\text{det}\left(\left(L(X')_{-1}\right)\right)=\sum_{T\in\mathcal{T}_{G'}} \prod_{f\in T} x'_f \geq nR^{n-1}.$$ Combining with Equation \ref{eq:det1} we have:

\begin{align*}
 R^{n-1}\text{det}\left(\left(L(X)_{-1}\right)\right)
 &=  \text{det}\left(\left(L(X')_{-1}\right)\right) \\
 & \geq nR^{n-1}.
\end{align*}
That is, $$\text{det}\left(\left(L(X)_{-1}\right)\right)\geq n$$ and the matrix-tree theorem implies $$\sum_{T\in\mathcal{T_G}} \prod_{e\in T} x_e \geq n.$$ \hfill
\end{proof}

We can now show that the matrix-tree theorem constraint (\ref{eq:MTT}) holds for any feasible point of the subtour LP.  We restate our main theorem in slightly more detail.

\begin{thm*}[Theorem \ref{thm:Main1}, restated]
Let $x\in \R^E$ be a feasible solution to the subtour LP (\ref{eq:SLP}) and let $G$ be the complete graph.  Let $X$ be the symmetric matrix where $X_{ij}=X_{ji}=x_{\{i, j\}}$ and $X_{ii}=0$ for all $i$.  Then $X$ satisfies the matrix-tree theorem constraint:
 $$\text{det}\left(\left(2I-X\right)_{-1}\right) \geq n.$$   Equivalently, $$\sum_{T\in \mathcal{T}_{G}} \prod_{e\in T} x_e \geq n.$$
\end{thm*}

\begin{proof} The subtour LP is bounded, so that every feasible $x$ for the subtour LP can be written as a convex combination of extreme points to the subtour LP.  For any extreme point of the subtour LP $y$, let $Y$ be the matrix where $Y_{ij}=Y_{ji}=y_{\{i, j\}}$ and $Y_{ii}=0.$  Feasibility of the subtour LP means $y(\delta(i))=2$ for all $i\in V$, so the associated Laplacian is $2I-Y$.  By Proposition \ref{thm:MT2} and the matrix-tree theorem, $\text{det}\left(\left(2I-Y\right)_{-1}\right) \geq n.$  

We now show that any convex combination of two extreme points of the subtour LP also satisfies the matrix-tree theorem constraint; extending to general convex combinations is left as an exercise.  Note that the determinant is well-known to be log concave on symmetric positive definite matrixes (see, e.g., section 3.1 of Boyd and Vandenberghe \cite{Boyd04})  so that $\text{det}(tA+(1-t)B)\geq \text{det}(A)^t \text{det}(B)^{1-t}$ for $0\leq t\leq 1$ if $A, B\succ 0.$  

Consider two extreme points of the subtour LP, with weighted adjacency matrices $A$ and $B$.  Denote their graph Laplacians as $L(A)=2I-A$ and $L(B)=2I-B$ respectively.  For a graph with weighted adjacency matrix $X$, all principal subminors of $L(X)$ are nonnegative so that all principal subminors of $L(X)_{-1}$ are as well: these are just the principal subminors of $L(X)$ that include row/column 1 being removed.  This implies that $L(X)_{-1}\succeq 0$.  By Proposition  \ref{thm:MT2}, $\text{det}\left(\left(L(A)_{-1}\right)\right), \text{det}\left(\left(L(B)_{-1}\right)\right)  \geq n$ so that zero cannot be an eigenvalue of $\left(L(A)_{-1}\right)$ or $\left(L(N)_{-1}\right)$ and so both are positive definite.  By log-concavity,
\begin{align*}
\text{det}&(t\left(L(A)_{-1}\right)+(1-t)\left(L(B)_{-1}\right)) \hspace{5mm} 
\\ &\geq 
\left(\text{det}\left(L(A)_{-1}\right)\right)^t \left(\text{det}\left(L(B)_{-1}\right)\right)^{1-t} \\
&\geq n^t n^{1-t} \\
&= n.
\end{align*}
Hence, $tA+(1-t)B$ satisfies the matrix-tree-theorem constraint (\ref{eq:MTT}). 
\hfill
\end{proof}

\begin{rem}
Note that the proof of Theorem \ref{thm:Main1} for any $x$ such that $x(\delta(S))\geq 2$ for each $S\subset V$ with $1\leq |S|\leq |V|-1.$  Hence, any $x$ with $x(\delta(S))\geq 2$ for all such $S$ and corresponding weighted adjacency matrix $X$ satisfies $$\text{det}\left(\left(L(X)_{-1}\right)\right) =\sum_{T\in \mathcal{T}_{G}} \prod_{e\in T} x_e \geq n.$$  However, it need not be the case that that rows of sum to $X$, so possibly $L(X)\neq 2I-X.$
\end{rem}

\section{Conclusion}

Theorem \ref{thm:Main1} has several implications.  Goemans and Rendl \cite{Goe00} show that the subtour LP is stronger than a TSP SDP relaxation of Cvetkovi{\'c} et al.\ \cite{Cvet99} in the following sense: Any feasible solution for the subtour LP corresponds to a feasible solution of the same cost for the SDP.  Hence, on any given instance, the optimal value of the subtour LP is at least as close to the cost of a TSP solution as the optimal value of the SDP.  Theorem \ref{thm:Main1} gives a comparable weakness result for the matrix-tree theorem constraint (\ref{eq:MTT}).  Moreover, it implies that Goemans and Rendl \cite{Goe00}'s result extends to the case where the matrix-tree theorem constraint (\ref{eq:MTT}) is added to the SDP of Cvetkovi{\'c} et al.\ \cite{Cvet99}.  More generally, our results show that matrix semidefinite inequalities used to impose the matrix-tree theorem are implied by a set of linear inequalities.

\section*{Acknowledgments}
This work was supported by NSF grant CCF-1908517.  This material is also based upon work supported by the National Science Foundation Graduate Research Fellowship Program under Grant No. DGE-1650441. Any opinions,
findings, and conclusions or recommendations expressed in this material are those of the
authors and do not necessarily reflect the views of the National Science Foundation.

\bibliography{bibliog} 
\bibliographystyle{abbrv}

\end{document}